\documentclass[letterpaper,twoside,twocolumn]{IEEEtran}

\IEEEoverridecommandlockouts                              % 
\usepackage{graphicx,float}
\usepackage{url}
\usepackage{amssymb,amsmath,amsfonts,layout}
\usepackage{amsthm} 
\usepackage{makeidx}
\usepackage{blkarray,multirow}
\usepackage{cite}
\usepackage{enumerate}
\usepackage{algpseudocode}
\usepackage{algorithm}
\usepackage{braket}
\usepackage{tikz}
\usepackage{lipsum}

\usepackage{ifthen,calc}
\usepackage{xstring}

\usetikzlibrary {positioning}
\usepackage{empheq}
\usepackage[most]{tcolorbox}

\newtcbox{\mybox}[1][]{%
    nobeforeafter, math upper, tcbox raise base,
    enhanced, colframe=blue!30!black,
    colback=blue!30, boxrule=1pt,
    #1}

\algnewcommand{\Inputs}[1]{%
	\State \textbf{Inputs:}
	\Statex \hspace*{\algorithmicindent}\parbox[t]{.8\linewidth}{\raggedright #1}
}
\algnewcommand{\Initialize}[1]{%
	\State \textbf{Initialize:}
	\Statex \hspace*{\algorithmicindent}\parbox[t]{.8\linewidth}{\raggedright #1}
}
\algnewcommand{\Outputs}[1]{%
	\State \textbf{Outputs:}
	\Statex \hspace*{\algorithmicindent}\parbox[t]{.8\linewidth}{\raggedright #1}
}

\usepackage{tikz}
\usetikzlibrary{shapes,arrows,positioning} 
\tikzset{
	>=stealth',
	block/.style={
		rectangle,
		rounded corners,
		draw=black, very thick,
		text width=12em,
		minimum height=3em,
		text centered},
	link/.style={
		->,
		thick,
		shorten <=2pt,
		shorten >=2pt},
	decision/.style={
		diamond,
		draw, very thick,
		fill=blue!20, 
		text width=8em,
		aspect=3,
		text centered}
}

\theoremstyle{plain}

\newtheorem{theorem}{Theorem}[section]

\newtheorem{lemma}[theorem]{Lemma}

\newtheorem{remark}[theorem]{Remark}

\newcommand{\longthmtitle}[1]{\mbox{}{\bf \textit{(#1).}}}

\newcommand{\oprocendsymbol}{\hbox{$\square$}}
\newcommand{\oprocend}{\relax\ifmmode\else\unskip\hfill\fi\oprocendsymbol}

\makeatletter
\newcommand{\pushright}[1]{\ifmeasuring@#1\else\omit\hfill$\displaystyle#1$\fi\ignorespaces}
\newcommand{\pushleft}[1]{\ifmeasuring@#1\else\omit$\displaystyle#1$\hfill\fi\ignorespaces}
\renewcommand*\env@matrix[1][*\c@MaxMatrixCols c]{%
	\hskip -\arraycolsep
	\let\@ifnextchar\new@ifnextchar
	\array{#1}}
\makeatother

\DeclareMathOperator{\xb}{\mathbf{x}}

\DeclareMathOperator{\zb}{\mathbf{z}}
\DeclareMathOperator{\wb}{\mathbf{w}}
\DeclareMathOperator{\Lb}{\mathbf{L}}

\DeclareMathOperator{\D}{\mathcal{D}}
\DeclareMathOperator{\Pint}{\mathbb{N}}

\DeclareMathOperator{\E}{\mathcal{E}}

\DeclareMathOperator{\graph}{\mathcal{G}}

\newcommand{\Nc}{\ensuremath{\mathcal{N}}}

\newcommand{\blkdiag}{\operatorname{blkdiag}}

\newcommand{\ones}{\mathbf{1}}

\newcommand{\real}{\ensuremath{\mathbb{R}}}

\def\@fnsymbol#1{\ensuremath{\ifcase#1\or *\or \dagger\or \ddagger\or
   \mathsection\or \mathparagraph\or \|\or **\or \dagger\dagger
   \or \ddagger\ddagger \else\@ctrerr\fi}}
\newcommand{\ssymbol}[1]{^{\@fnsymbol{#1}}}

\newcommand\norm[1]{\left\lVert#1\right\rVert}

\allowdisplaybreaks

\graphicspath{{epsfiles/}}

\begin{document}

\title{A Privacy Preserving Distributed Model Identification Algorithm for Power Distribution Systems}
\author{Chin-Yao Chang
\thanks{C.-Y. Chang is with the National Renewable Energy Laboratory, Golden, CO 80401, USA (Emails: \{chinyao.chang\}@nrel.gov).}
\thanks{This work was authored in part by NREL, operated by Alliance for Sustainable Energy, LLC, for the U.S. Department of Energy (DOE) under Contract No. DE-AC36-08GO28308. Funding provided by DOE Office of Electricity, Advanced Grid Modeling Program, through agreement NO. 33652. The views expressed in the article do not necessarily represent the views of the DOE or the U.S. Government. The U.S. Government retains and the publisher, by accepting the article for publication, acknowledges that the U.S. Government retains a nonexclusive, paid-up, irrevocable, worldwide license to publish or reproduce the published form of this work, or allow others to do so, for U.S. Government purposes.}
}

\maketitle

\begin{abstract}
Distributed control/optimization is a promising approach for network systems due to its advantages over centralized schemes, such as robustness, cost-effectiveness, and improved privacy. However, distributed methods can have drawbacks, such as slower convergence rates due to limited knowledge of the overall network model. Additionally, ensuring privacy in the communication of sensitive information can pose implementation challenges. To address this issue, we propose a distributed model identification algorithm that enables each agent to identify the sub-model that characterizes the relationship between its local control and the overall system outputs. The proposed algorithm maintains the privacy of local agents by only communicating through dummy variables. We demonstrate the efficacy of our algorithm in the context of power distribution systems by applying it to the voltage regulation of a modified IEEE distribution system. The proposed algorithm is well-suited to the needs of power distribution controls and offers an effective solution to the challenges of distributed model identification in network systems. 

% The sub-models can be combined to the full system model for centralized control use if needed. 
%The sub-models are updated online with a new batch of input-output data, making the identified model adapt well with time-varying systems. 
\end{abstract}

\section{Introduction}
In recent years, there has been a growing trend towards emphasizing the privacy of users. This has led to the adoption of stricter regulations and policies to protect user privacy, such as the General Data Protection Regulation (GDPR) in Europe~\cite{voigt2017eu} and the California Consumer Privacy Act (CCPA)~\cite{bukaty2019california} in the United States. As the trend towards emphasizing privacy is likely to continue, how to extract values from data without sacrificing privacy becomes very relevant in many applications such as healthcare, insurance, and FinTech. The technology behind is known as federated learning, which has found success in language model, image recognition, etc~\cite{konevcny2016federated,li2020review}. Some recent works~\cite{wang2021electricity,lin2021privacy} also found that federated learning can be useful in developing energy management strategies for future power grids, which could involve millions of controllable devices~\cite{kroposki2020aeg}. However, those are restricted to gaining knowledge of the pattern of certain classes of power consumption or generation, while can not help identify a system-level model (e.g. power flow model), which is very much needed for sophisticated control of future power distribution systems.

Distribution systems are mostly unobservable largely because of the cost-ineffectiveness of installing all the sensors. Many existing works are then on identifying some valuable information for grid controls, such as topology of the distribution system~\cite{bolognani2013identification,ardakanian2019identification} and state estimation~\cite{primadianto2016review,de2016improved,guo2020solving}. Though distribution system state estimation (DSSE) is about estimating the voltage and currents for buses without sensors, many DSSE still requires some knowledge of the admittance matrix or similar information~\cite{de2016improved,guo2020solving}. The knowledge of the admittances is also valuable for various controls of distribution systems, e.g.,~\cite{salim2012model,chang2019saddle}, but such information is not necessarily easy to obtain, especially considering the data collection hurdles when privacy comes into the equation. 

The aforementioned works assume a central entity collects the data from all the sensors in the distribution system for the estimations, which is also the case for estimating the admittance matrix~\cite{zhang2021distribution}. Collecting all the available data, especially the power consumption patterns of the local users or distributed energy resources can be challenging. Generally speaking, there are some distributed algorithms that can estimate the system model, say the admittance matrix, through distributed communication~\cite{stegagno2019distributed,huang2022scalable}. However, distributed algorithms mechanically have a consensus element that are usually about the consensus of the estimated models or even the local state variables of the agents, which leaves the valuable information floating on the communication network that could be susceptible for cyber attacks. Another downside of many distributed algorithms is that they are essentially gradient descent that makes them robust over package drops or delays, while the convergence rate is compromised. Adam algorithm (Adaptive Moment Estimation)~\cite{kingma2014adam} provides sophisticated step-size (learning rate) selection that can improve the convergence rate of the distributed algorithms. Adam algorithm is also found very effective for large-scale optimizations~\cite{chilimbi2014project,jose2016scalable}, thus it grows as one of the mainstream algorithms in the machine learning field~\cite{ruder2016overview}. Bringing in Adam algorithm elements in distributed algorithms can potentially improve the convergence rate.

%Identification of the model~\cite{papusha2014collaborative}, while the model is in the form $[y_1;\cdots;y_n] = blkdiag(\{\theta_i\}_{i=1}^n) [u_1;\cdots;u_n]$
%\cite{bolognani2013identification} presents a method for identifying the topology of a power distribution network based on voltage correlation analysis.
%\cite{ardakanian2019identification} proposed framework that leverages smart meters and phasor measurement units (PMUs) data to identify the topology of the distribution grid and estimate the state of the grid with the presence of measurement errors and communication delays.
%Multiagent systems: robotic, economics, telecommunications 
%Multiagent reinforcement learning~\cite{busoniu2008comprehensive}
%Distributed Adam~\cite{li2022federated}

\textit{Contributions:} The contribution of the paper is mainly on developing a distributed model identification algorithm such that each agent in the network system identifies the sub-model for the correlation between its local controls and the overall system outputs, which can be understood as a sub-matrix of the LinDistFlow model (or its equivalence) for the distribution system if the inputs are the power injections and outputs are the voltage magnitudes. The proposed algorithm takes elements of Adam algorithm to improve the convergence rate. On top of it, it has an appealing property that the packets exchanged between the agents are dummy variables so that local data and the identified model are kept private. We demonstrate the effectiveness of the proposed distributed model identification algorithm with a modified IEEE test system.

\section{Preliminary}\label{sec:prelim}
\subsection{Notations}
For $x\in\real^n$, we denote its ${l}^2$-norm and quadratic norm in terms of matrix $A\succ 0$ by $\|x\|_2$ and $\|x\|_A$, respectively. For matrices $A_1,A_2,\cdots,A_N$ with $A_i\in\real^{m\times n}$ for all $i=1,2,\cdots,N$, we denote $\blkdiag{(\{A_i\}_{i=1}^{N})}$ as the block diagonal matrix of all the $A_i$; $[A_1;A_2;\cdots; A_n]\in\real^{Nm \times n}$ and $[A_1,A_2,\cdots, A_n]\in\real^{m\times Nn}$ respectively as the vertical and horizontal concatenations. The Hadamard product and division of $A_i$ and $A_j$ are denoted as $A_i \odot A_j$ and $A_i ./ A_j$, respectively. For a matrix $A\in\real^{m\times n}$, $vec(A)\in\real^{mn}$ is a column vector that vectorizes $A$ by concatenating the column vectors of $A$ from left to right. The kernel of matrix $A$ is denoted as ker$(A)$. The scalar element of the $i^{th}$ row and $k^{th}$ column of $A$ is denoted as $A(i,k)$; $A(:,k)$ is the $k^{th}$ column vector of $A$. 
%The maximal and minimal eigenvalues of $A$ are given as $\lambda_M(A)$ and $\lambda_m(A)$, respectively. 
An identity matrix and all-ones vector with dimension $n$ are denoted as $I_n$ and $\ones_n$, respectively. 

\subsection{System modeling with input-output data}\label{sec:model}
%Our recent work~\cite[Section II]{chang2022robust} shows that for many classes of static/dynamical systems, the data representation of the system model can be written into a quite compact unified form. 
In the following, we briefly overview the data-driven modeling framework which will be used for the remainder of this paper. Let $u(k)\in\real^m$ and $y(k)\in\real^n$ be the input and output data at time instant $k$ for an unknown system. Given the input-output data collected for time $k=0,\cdots,T-1$, we define the data matrices shown in the following:
\begin{subequations}\label{eq:data_matrices}
\begin{align}
    &U = [\phi_u(u(0)),\cdots,\phi_u(u(T-1))], \\ 
    &Y = [\phi_y(y(0)),\cdots,\phi_y(y(T-1))],
\end{align}
\end{subequations}
where $\phi_u: \real^m \mapsto \real^m$ and $\phi_y: \real^n \mapsto \real^n$ are the mappings that capture the known (can be nonlinear) physics of the system to be identified. Assuming that the only unknown part of the targeted system is linear and characterized by $A\in\real^{n\times m}$, the following equation holds:
\begin{align} \label{eq:data_lin_eq}
    Y = A U.
\end{align}
The model identification problem with given data ($U$ and $Y$) can be understood as solving the linear equation~\eqref{eq:data_lin_eq}. If $U$ has full row rank, then the system model $A$ is uniquely defined. The straight data representation of the system model with data~\eqref{eq:data_lin_eq} is useful for controller design purposes, more details are available in~\cite{de2019formulas,chang2022robust}.

\section{Distributed model identification algorithm}\label{sec:dis_model_ID}

In this section, we will first show how~\eqref{eq:data_lin_eq} is formulated in a network system, and then lay out how to reformulate it as a distributed optimization problem. We next develop a distributed model identification algorithm leveraging Adam adaptive step-size. 
%The locally identified sub-models enable a decentralized control scheme such that the performance equates centralized controls.
% leveraging Adam for solving the distributed optimization with an upper bound of the regret.

\subsection{System modeling of network systems}
% For a network system, the actuators and sensors are spread in various locations of the interconnected network. 
We consider a network system that is partitioned into $N$ number of regions, and each region has an associated agent that collect all the actuator and sensor data in the region. We assume such a network partitioning setup effectively makes each agent $i\in\Nc:=\{1,2,\cdots,N\}$ only be able to get information on certain rows of $U$ and $Y$ for the purpose of formulating model identification problem~\eqref{eq:data_lin_eq}. Define $\D_{u,i}$ and $\D_{y,i}$ as the sets of rows of $U$ and $Y$ that are known for agent $i$, respectively. We assume the full observability of the system in the sense that $\cup_{i\in\Nc}\D_{u,i} = \{1,2,\cdots,m\}$ and $\cup_{i\in\Nc}\D_{y,i} = \{1,2,\cdots,n\}$ with $\D_{u,i}\cap \D_{u,j} = \emptyset$ and $\D_{y,i}\cap \D_{y,j} = \emptyset$ for all $i\not= j$ for simplicity. In this setup, the agent $i$ can capture how its regional controls affect the output $Y$ by knowing just some columns of $A$ instead of the full matrix, for which we define the sub-matrix by $A_{\D_{u,i}} \in \real^{n\times |\D_{u,i}|}$ and without loss of generality, $A = [A_{\D_{u,1}},A_{\D_{u,2}},\cdots,A_{\D_{u,N}}]$. Note that even only $A_{\D_{u,i}}$ is needed for agent $i$, it is achievable to identify and reach consensus on the full model $A$ for all the agents; however, that can lead to massive data exchanges between the agents that overwhelm the distributed communication network. Therefore, the goal for each agent $i$ is identifying $A_{\D_{u,i}}$ by distributed communication and locally available sub-matrices of $U$ and $Y$.

%Let $U(k)$ and $Y(k)$ be the $k^{th}$ column of the data matrix $U$ and $Y$, respectively (or $k^{th}$ data point). Without loss of generality, we label the locally available elements of the input-output data of the agent $i$ by $(U_{i}(k),Y_{i}(k))$ such that $U(k) = \sum_{i=1}^N U_{i}(k)$ and $Y_k = \sum_{i=1}^N Y_{i}(k)$, where every vector $U_{i}(k)\in\real^m$ and $Y_{i}(k)\in\real^n$ are sparse and only takes non-zero values if those elements are associated with the measurements that agent $i$ has access to. In this setup,~\eqref{eq:data_lin_eq} is rewritten as  
%\begin{align} \label{eq:data_lin_eq_2}
%    Y = \sum_{i=1}^N Y_i = A \sum_{i=1}^N U_i,
%\end{align}
%where $U_i = [U_i(0),\cdots,U_i(T-1)]$ and $Y_i = [Y_i(0),\cdots,Y_i(T-1)]$. We further partition $A = [A_1,\cdots,A_N]$ such that each $A_i\in\real^{n\times m_i}$ collects the columns associated with non-zero elements of $U_i(k)$ and rewrite~\eqref{eq:data_lin_eq_2} as 
%\begin{align} \label{eq:data_lin_eq_3}
%    Y = \sum_{i=1}^N Y_i = \sum_{i=1}^N A_i U_i. 
%\end{align}
%Equation~\eqref{eq:data_lin_eq_3} indicates that from the control perspective, each agent $i$ only needs to know the sub-model $A_i$ because $A_j$, $j\not=i$, are not relevant to the controls of agent $i$, $U_i$, anyway. 

\subsection{Distributed reformulation of the system modeling}\label{sec:dis_reform}
In this section, we go through a series of  reformulations of~\eqref{eq:data_lin_eq} for convenience of distributed algorithm design. Define $\xb_i= vec(A_{\D_{u,i}})$ and $\xb= [\xb_1;\xb_2;\cdots;\xb_n] \in\real^{nm}$. In this rearrangement, we consider the following formulation of~\eqref{eq:data_lin_eq}:
\begin{align}\label{eq:data_opt}
    &\min_{\xb} \frac{1}{2}\sum_{k=1}^T \norm{\Big(\sum_{i=1}^N U_{k,\D_{u,i}}\cdot \xb_i \Big) - Y(:,k)}^2_2,
    %&\text{ s.t. } \xb_1 = \xb_2 = \cdots = \xb_N,  \label{eq:data_lin_eq_2-2}
\end{align}
where $U_{k,\D_{u,i}}\in\real^{n\times n\cdot|\D_{u,i}|}$ is the horizontal concatenation of $U(i,k)\cdot I_n$ for all $i\in\D_{u,i}$. One can verify that the optimal solution of~\eqref{eq:data_opt} is a solution of~\eqref{eq:data_lin_eq} by direct algebra. The reason for the optimization formulation is that practically, it is unlikely to find $\xb$ such that~\eqref{eq:data_lin_eq} holds due to the noisy data, communication disturbances or other disturbances. We next rewrites~\eqref{eq:data_opt} in a more compact way.  Defining $U_{\D_{u,i}} = [U_{1,\D_{u,i}};U_{2,\D_{u,i}};\cdots;U_{T,\D_{u,i}}]$, $Y_{\D_{y,i}} = [Y_{1,\D_{y,i}};Y_{2,\D_{y,i}};\cdots;Y_{T,\D_{y,i}}]$ with $Y_{k,\D_{y,i}}\in\real^n$ such that
\begin{align*}
    Y_{k,\D_{y,i}}(j) = \begin{cases} Y(i,k) & \text{if } j\in\D_{y,i}, \\  0 & \text{otherwise, }\end{cases}
\end{align*}
we rewrite~\eqref{eq:data_opt} as
%$\hat{Y}_i = [Y_i(1);Y_i(2);\cdots;Y_i(T)]$ and $\hat{U}_i = [I_n \otimes U_{i}(1)^\top;\cdots;I_n \otimes U_{i}(T)^\top]$ for all $i=1,\cdots,N$, we next rewritten~\eqref{eq:data_opt} in a compact way:
\begin{align}\label{eq:data_opt_2}
    &\min_{\xb} f(\xb), \quad f(\xb):= \frac{1}{2} \norm{\sum_{i=1}^N (U_{\D_{u,i}} \xb_i - Y_{\D_{y,i}})}^2_2.
    %&\text{ s.t. } \xb_1 = \xb_2 = \cdots = \xb_N,  \label{eq:data_lin_eq_2-2}
\end{align}
The objective function $f(\xb)$ couples the variables and data for all the agents, which is not yet solvable with distributed communications. In the following, we first leverage the results in~\cite{huang2022scalable} for an algorithm that can solve~\eqref{eq:data_opt_2} distributively, followed with some modifications with Adam step-size and illustrating the privacy preserving properties. Define the graph associated with the distributed communication network as $\graph= (\Nc,\E)$, where  $\E\subseteq\Nc\times \Nc$ is the set of edges (communication links). We next consider the following optimization:
\begin{align}\label{eq:data_opt_3}
    \min_{\xb,\wb} \hat{f}(\xb,\wb), \quad \hat{f}(\xb,\wb):= \frac{1}{2} \norm{\hat{U} \xb - \hat{Y} - \Lb^{\frac{1}{2}}\wb}^2_2,
\end{align}
where $\wb\in\real^{TNn}$ is a newly introduced slack variable, $\hat{U} = \blkdiag{(\{\hat{U}_{\D_{u,i}}\}_{i=1}^{N})}$, $\hat{Y} = [Y_{\D_{y,1}}; Y_{\D_{y,2}}, \cdots , Y_{\D_{y,N}}]$,   $\Lb = (L\otimes I_{nT})$, $L$ is the Laplacian matrix associated with $\graph$, and $\Lb^{\frac{1}{2}}$ is the square root of $\Lb$. Lemma~\ref{Lem:eq_opt3} (a compact version of~\cite[Lemma 3.1]{huang2022scalable}) shows that the optimal solutions of~\eqref{eq:data_opt_3} are also the ones for~\eqref{eq:data_opt_2} through KKT optimality condition arguments. 
\begin{lemma}\longthmtitle{Optimal solutions of~\eqref{eq:data_opt_2} and~\eqref{eq:data_opt_3}}
\label{Lem:eq_opt3}
If $\graph$ is connected and $(\xb^\star,\wb^\star)$ is an optimal solution of~\eqref{eq:data_opt_3}, then $\xb^\star$ is an optimal solution of~\eqref{eq:data_opt_2}.
\end{lemma}
\proof
By KKT conditions, $(\xb^\star,\wb^\star)$ is an optimal solution of~\eqref{eq:data_opt_3} if and only if
\begin{subequations}\label{eq:kkt_opt_3}
    \begin{align}
        \hat{U}^\top (\hat{U} \xb^\star - \hat{Y} - \Lb^{\frac{1}{2}}\wb^\star) = 0, \label{eq:kkt_opt_3-1} \\
        {\Lb^{\frac{1}{2}}}^\top(\hat{U} \xb^\star - \hat{Y} - \Lb^{\frac{1}{2}}\wb^\star) = 0. \label{eq:kkt_opt_3-2}
    \end{align}
\end{subequations}
By~\eqref{eq:kkt_opt_3-2} and the property of ker$(\Lb^{\frac{1}{2}})=$ ker$(\Lb)$, we can define $\zb^\star = \ones_N \otimes z^\star = \hat{U} \xb^\star - \hat{Y} - \Lb^{\frac{1}{2}}\wb^\star$. We derive the following equation by left multiplying $\zb^\star$ by $\ones^\top_N \otimes I_{nT}$:
\begin{align}\label{eq:Nz_star}
    N z^\star = \sum_{i=1}^N(\hat{U}_{\D_{u,i}} \xb^\star_i - \hat{Y}_{\D_{y,i}}),
\end{align}
where ker$(\Lb^{\frac{1}{2}})=$ ker$(\Lb)$ is used again in deriving~\eqref{eq:Nz_star}. Substituting~\eqref{eq:Nz_star} to~\eqref{eq:kkt_opt_3-1} gives
\begin{align}\label{eq:kkt_opt_2}
    \hat{U}_{\D_{u,j}}^\top \Big(\sum_{i=1}^N(\hat{U}_{\D_{u,i}} \xb^\star_i - \hat{Y}_{\D_{y,i}}) \Big) = 0, \quad\quad \forall j\in\Nc.
\end{align}
Because~\eqref{eq:kkt_opt_2} is the KTT condition for the optimality of optimization~\eqref{eq:data_opt_2}, we conclude that $\xb^\star$ is also an optimal solution of~\eqref{eq:data_opt_2} and complete the proof.  \qed 
%\textit{Necessity}: Given the optimal solution of~\eqref{eq:data_opt_2}, $\xb^\star$, we define $z^\star$ and $\zb^\star$ such that $z^\star = \frac{1}{N}\Big(\sum_{i=1}^N(\hat{U}_i \xb^\star_i - \hat{Y}_i) \Big)$ and $\zb^\star = \ones_N \otimes z^\star$. Because $(\ones^\top_N \otimes I_{Tn})\zb^\star = (\ones^\top_N \otimes I_{Tn})(\hat{U} \xb^\star - \hat{Y})$, there exists a $\wb^\star$ such that
%\begin{align}
%    \hat{U} \xb^\star - \hat{Y} - \zb^\star - \Lb^{\frac{1}{2}}\wb^\star = 0.
%\end{align}
\endproof
%Optimization~\eqref{eq:data_opt_3} can be solved in a distributed way as shown in~\cite{huang2022scalable}. Because the formulation in~\cite{huang2022scalable} is unnecessary complicated in solving~\eqref{eq:data_opt_3}, we will go through the key steps of the distributed algorithm for~\eqref{eq:data_opt_3} in the next section.
% Because our model identification problem~\eqref{eq:data_opt_2} can be changed due to different batch of data as opposed to unchanged targeted optimization in~\cite{huang2022scalable}, we take a different approach from the gradient method proposed in~\cite{huang2022scalable}, detailed in the next section.
With Lemma~\ref{Lem:eq_opt3}, we focus on solving~\eqref{eq:data_opt_3} distributively for the original model identification problem.

\subsection{Distributed algorithm for model identification}
The first step of the distributed algorithm for~\eqref{eq:data_opt_3} is looking into the gradient descent of~\eqref{eq:data_opt_3}:
\begin{subequations}
\label{eq:grad}
\begin{align}
   \dot{\xb} &= -\hat{U}^\top (\hat{U}\xb - \hat{Y} - \Lb^{\frac{1}{2}} \wb), \label{eq:grad-1} \\
   \dot{\wb} &= \Lb^{\frac{1}{2}}(\hat{U}\xb - \hat{Y} - \Lb^{\frac{1}{2}} \wb). \label{eq:grad-2}        
\end{align}
\end{subequations}
Because $\Lb^{\frac{1}{2}}$ is not necessarily sparse,~\eqref{eq:grad} can not be directly implemented in a distributed way. Even if we formulate~\eqref{eq:data_opt_3} by replacing $\Lb^{\frac{1}{2}}$ with $\Lb$ (the results of Lemma~\ref{Lem:eq_opt3} still hold), the associated gradient descent with $\Lb^{\frac{1}{2}}$ replaced by $\Lb$ in~\eqref{eq:grad} still requires packet exchanges between two-hop neighbors, and the local input-output data should be shared between the agents, which is not desirable. By introducing a change of variable 
\begin{align} \label{eq:z_def}
\zb= \hat{U} \xb - \hat{Y} - \Lb^{\frac{1}{2}}\wb,    
\end{align}
we instead consider the following alternative gradient method:
\begin{subequations}\label{eq:pd}
\begin{align}
    \dot{\xb} &=-\hat{U}^\top \zb,  \\
    \dot{\zb} &= - \Lb \zb + \hat{U} \dot{\xb} = - (\Lb + \hat{U}\hat{U}^\top) \zb.
\end{align} 
\end{subequations}
Because $\hat{U}$ is block diagonal with the off-diagonal elements being zeros, the implementation of~\eqref{eq:pd} only requires distributed communication of $\zb$ with $\zb$ being partitioned properly in a way that $\zb = [\zb_1; \zb_2;\cdots;\zb_N]$, $\zb_i \in\real^{nT}$. \cite{huang2022scalable} has shown that~\eqref{eq:pd} and its Euler discrete formulation~\eqref{eq:pd_dis} converge to the optimal solution of~\eqref{eq:data_opt_3}.
\begin{subequations}\label{eq:pd_dis}
\begin{align}
    \xb(k+1) &= \xb(k) - \alpha\hat{U}^\top \zb(k),  \\
    \zb(k+1) &= \zb(k) - \alpha\Big(\Lb  + \hat{U}\hat{U}^\top\Big)\zb(k).
\end{align} 
\end{subequations}
To improve the convergence rate of~\eqref{eq:pd_dis}, we propose to introduce Adam adaptive step-size to~\eqref{eq:pd_dis}. 
Defining $D = \Lb  + \hat{U}\hat{U}^\top$ and $N_I$ as the number of iterations for the model identification, we propose the distributed model identification algorithm with Adam adaptive step-size in Algorithm~\ref{alg:D-ADAM}.
%\footnote{The square root of $\hat{s}_2(k)$ in line~\ref{alg:x} is element-wise}
\begin{algorithm}
\caption{Distributed Model Identification Algorithm with Adam adaptive step-size}\label{alg:D-ADAM}
\begin{algorithmic}[1]
\Require $\beta_1, \beta_2 \in [0,1)$, $\epsilon >0$, $k=0$
\Require Initialize $s_1(0)= s_2(0) = \hat{s}_1(0)= \hat{s}_2(0)= 0$, and $\xb(0)$, $\wb(0)$, $\zb(0)$ that satisfies~\eqref{eq:z_def}, 
\While{$k \leq N_I-1$}
    \State $k \gets k+1$ 
    \State $g(k) \gets \hat{U}^\top \zb(k-1)$
    \State $s_1(k) \gets \beta_1 s_1(k-1) + (1 - \beta_1)g(k) $ \label{alg:s1}
    \State $s_2(k) \gets \beta_2 s_2(k-1) + (1 - \beta_2) \Big( g(k) \odot g(k) \Big)$
    \State $\hat{s}_1(k) \gets s_1(k)/ (1-\beta_1^k)$
    \State $\hat{s}_2(k) \gets s_2(k)/ (1-\beta_2^k)$
    \State $\xb(k) \gets \xb(k-1) - \alpha \Big(\hat{s}_1(k)./(\sqrt{\hat{s}_2(k)} + \epsilon) \Big)$ \label{alg:x}
    \State $\zb(k) \gets \zb(k-1) - \alpha D\zb(k-1)$  \label{alg:z} % \Comment{This is a comment}
\EndWhile
\end{algorithmic}
\end{algorithm}

Note that the adaptive step-size is only used on the $\xb$ dynamics in line~\ref{alg:x} but not on $\zb$ in line~\ref{alg:z} because such a adaptive step-size on $\zb$ can break the consensus established through $\Lb$. Although we do not conclude a faster convergence rate in the proof of convergence of Algorithm~\ref{alg:D-ADAM} stated in Theorem~\ref{thm:converge}, we expect that Algorithm~\ref{alg:D-ADAM} converges faster than~\eqref{eq:pd_dis} due to the outperformance of Adam algorithm over the gradient descent method generally~\cite{ruder2016overview}.

\begin{theorem}\longthmtitle{Convergence of Algorithm~\ref{alg:D-ADAM}}\label{thm:converge}
Algorithm~\ref{alg:D-ADAM} has $\xb(k)\to \xb^\star$ as $k\to \infty$ if $\|\hat{U}\|_2<\infty$, $\|\hat{Y}\|_2<\infty$, $\beta_1, \beta_2 \in [0,1)$ and $\alpha$ is chosen such that
%\begin{subequations}
%\label{eq:converge_cond}
\begin{align}
    &(\alpha^2 + \mu) D^\top P D - \alpha\Big( D^\top P + P D \Big) \preceq 0 \label{eq:converge_cond-1} 
 %   \\ &\frac{\beta^2_1}{\sqrt{\beta_2}} < 1. \label{eq:converge_cond-2}
\end{align}
%\end{subequations}
for some $P\succ 0$ and $\mu>0$.
\end{theorem}
\begin{proof}
With line~\ref{alg:z} of Algorithm~\ref{alg:D-ADAM},~\eqref{eq:converge_cond-1} implies
\begin{align} \label{eq:converge_cond-1_re}
    &\zb^\top(k) P \zb(k) - \zb^\top(k-1) P \zb(k-1) \nonumber \\ 
    &\hspace{10mm}\leq -\mu\zb^\top(k-1) D^\top P D \zb(k-1), \nonumber \\
    \Rightarrow &\|\zb(k)\|_P- \|\zb(k-1)\|_P \leq -\mu\|D\zb(k-1)\|_P.
\end{align}
We next sum up~\eqref{eq:converge_cond-1_re} for $k=1,\cdots,N_I$, leading to
\begin{align} \label{eq:converge_cond-1_re_2}
    &\|\zb(N_I)\|_P \leq \|\zb(0)\|_P  - \sum_{k=0}^{N_I-1}\mu\|D\zb(k)\|_P.
\end{align}
Inequality~\eqref{eq:converge_cond-1_re_2} implies that $\|D\zb(k)\|_P\rightarrow 0$ as $k\rightarrow\infty$. Because $D= \Lb +\hat{U}\hat{U}^\top$, and $\Lb$ and $\hat{U}\hat{U}^\top$ are positive semidefinite, for any $\zb(k)\not=0$, $\|D\zb(k)\|_P=0$ if and only if $\Lb \zb(k) = \hat{U}\hat{U}^\top \zb(k) = 0$. Recall that $\zb^\star:= \hat{U}\xb^\star -\hat{Y}-\Lb^{\frac{1}{2}}\wb^\star$ is the optimal solution if and only if~\eqref{eq:kkt_opt_3} holds, so we can conclude $\zb(k) \rightarrow \zb^\star$ by $\|D\zb(k)\|_P \rightarrow 0$. Because $\norm{\zb^\star}_2< \infty$ and $\Lb^{\frac{1}{2}}\zb^\star = 0$, for any $\xb(k)$ with $\norm{\xb(k)}_2 <\infty$, there exists a $\wb(k)$ that solves
\begin{align*}
    \zb^\star= \hat{U} \xb(k) - \hat{Y} - \Lb^{\frac{1}{2}}\wb(k). 
\end{align*}
The $\xb(k)$ and $\wb(k)$ that solve the equation above are the optimal solutions of~\eqref{eq:data_opt_3}. In other words, we can conclude $\xb(k) \rightarrow \xb^\star$ given $\zb(k) \rightarrow \zb^\star$ as $k\rightarrow\infty$ by showing that $\norm{\xb(k)}_2<\infty$ for all $k\in\Pint$, or
\begin{align}\label{eq:x_bdd}
    \norm{\xb(k_1) - \xb(k_1)}_2 <\infty, \quad \forall \; k_1, \;k_2\in\Pint
\end{align}
To show~\eqref{eq:x_bdd} holds, we first state that there exists a $C_{g}<\infty$ such that 
\begin{align}\label{eq:bdd_sums-1}
\sum_{k=1}^{\infty}\|\hat{U}^\top\zb(k)\|_2 = \sum_{k=1}^{\infty} \|g(k)\|_2  \leq C_g.
\end{align}
Otherwise,~\eqref{eq:converge_cond-1_re_2} does not hold when $N_I\rightarrow \infty$. We next 
analyze the series of $\hat{s}_1(k)./(\sqrt{\hat{s}_2(k)} + \epsilon) $ that appears in the updates of $\xb(k)$ in Algorithm~\ref{alg:D-ADAM}:
\begin{align}
    &\norm{\sum_{k=1}^{\infty} \Big(\hat{s}_1(k)./(\sqrt{\hat{s}_2(k)} + \epsilon)\Big)}_2     \leq \frac{1}{\epsilon}\norm{\sum_{k=1}^{\infty}\hat{s}_1(k)}_2 \nonumber \\
    =& \frac{1}{\epsilon}\norm{\sum_{k=1}^{\infty}
    \frac{ \beta_1^{k}s_1(0) + (1-\beta_1)\sum_{i=1}^k\beta_1^{k-i} g(i) }{1-\beta_1^k}}_2, \label{eq:pf_1}
\end{align}
where line~\ref{alg:s1} of Algorithm~\ref{alg:D-ADAM} is used in deriving~\eqref{eq:pf_1}. Because $s_1(0) = 0$, we can simplify~\eqref{eq:pf_1} to:
\begin{align}
    & \norm{\sum_{k=1}^{\infty} \Big(\hat{s}_1(k)./(\sqrt{\hat{s}_2(k)} + \epsilon)\Big)}_2  \nonumber \\
    \leq&\frac{1-\beta_1}{\epsilon}\norm{\sum_{k=1}^{\infty} \frac{1}{1-\beta_1^k}\sum_{i=1}^k\beta_1^{k-i} g(i) }_2 \nonumber \\
    = &\frac{1-\beta_1}{\epsilon}\norm{\sum_{j=1}^{\infty} g(j) \cdot\Big(\sum_{h=j}^{\infty}\frac{\beta_1^{h-j}}{1-\beta_1^h}\Big) }_2 \nonumber \\
    \leq & \frac{1-\beta_1}{\epsilon}\cdot\frac{1}{(1-\beta_1)^2}\sum_{j=1}^{\infty} \|g(j)\|_2\leq \frac{C_g}{\epsilon(1-\beta_1)} <\infty \label{eq:pf_2}
\end{align}
The inequality~\eqref{eq:pf_2} implies~\eqref{eq:x_bdd} holds, which completes the proof.
  \qed 
%Recall that $\zb^\star:= \hat{U}\xb^\star -\hat{Y}-\Lb^{\frac{1}{2}}\wb^\star$ is the optimal solution if and only if~\eqref{eq:kkt_opt_3} holds, so we can conclude $\zb(k) \rightarrow \zb^\star$ by $\|D\zb(k)\|_P\rightarrow 0$. 

%It has been shown in~\cite[Theorem 10.5]{kingma2014adam} that if~\eqref{eq:converge_cond-2},~\eqref{eq:bdd_sums-1}, and~\eqref{eq:bdd_sums_x} hold, then for any $N_I\in\Pint$, there exists a $C(\alpha, \beta_1,\beta_2, N_I) <\infty$ such that
%\begin{align}
%    \sum_{k=0}^{N_I-1} (\xb(k) - \xb^\star)^\top U^\top \zb(k) \leq C(\alpha, \beta_1,\beta_2, N_I).
%\end{align}

\end{proof}
\begin{remark}\longthmtitle{Comparison of Algorithm~\ref{alg:D-ADAM} with Adam algorithm}
By dropping line~\ref{alg:z} of Algorithm~\ref{alg:D-ADAM} and changing the meaning of $g(k)$ from $\hat{U}^\top\zb(k-1)$ to the gradient of a certain objective function of $\xb$, Algorithm~\ref{alg:D-ADAM} is actually Adam algorithm. Such a change leads in very different routes to prove the convergence to the optimal $\xb^\star$. However, the fundamental assumptions are similar. The convergence statement of Adam algorithm~\cite[Theorem 10.5]{kingma2014adam} straight assumes~\eqref{eq:x_bdd} and \eqref{eq:bdd_sums-1} hold, and for Algorithm~\ref{alg:D-ADAM}, such an assumption is indirectly embedded in~\eqref{eq:converge_cond-1}, which actually leads to~\eqref{eq:x_bdd} and \eqref{eq:bdd_sums-1} as shown in the proof. Last but not least, the convergence statement of Adam algorithm~\cite[Theorem 10.5]{kingma2014adam} has additional assumptions on $\beta_1$ and $\beta_2$, which are used for deriving the bounds of the regret (or characterizing the convergence rate). Because Algorithm~\ref{alg:D-ADAM} does not use $\beta_1$ and $\beta_2$ for all the state variables (it only uses them in updating $\xb$), it is unclear on how to make additional assumptions on $\beta_1$ and $\beta_2$ to help characterize the convergence rate of Algorithm~\ref{alg:D-ADAM}. Practically, choosing $\beta_1$ and $\beta_2$ for Algorithm~\ref{alg:D-ADAM} based on numerical experiences of Adam algorithm in the literature works reasonably well.
\end{remark}

We conclude this section by noting that Algorithm~\ref{alg:D-ADAM} possesses some desirable properties: (1) only the dummy variables $\zb$ are exchanged through the communication network; (2) retrieving the sub-model, $\xb_i$ (or $A_{\D_{u,i}}$), that is useful for agent $i$ from $\zb$ requires local data $U_{\D_{u,i}}$ and $Y_{\D_{u,i}}$, which simultaneously preserves the privacy and enhances the security. We view those as the unique strengths of Algorithm~\ref{alg:D-ADAM} compared to other model identification algorithms. 

%\subsection{Decentralized and distributed controls with the identified sub-models}
%For the purpose of completeness, we discuss in this section on how the identified sub-models (sub-matrices) can be useful for decentralized or distributed controls with system-level awareness. Consider the target optimization that the agents want to solve collaboratively is given as
%\begin{align}\label{eq:quad_opt}
%    &\min_{u} \frac{1}{2}y^\top Q_2 y + Q_1 y, \;\;\text{ s.t. } y = A u,
%\end{align}
%where $Q_2 \succeq 0$. Optimization~\eqref{eq:quad_opt} is non-sparse and generally requires knowing all the $y$ to compute the gradient descent steps . 
%If $Q$ is known for all the agents, the sub-matrices $A_i$ for all $i\in\Nc$ are identified with Algorithm~\ref{alg:D-ADAM} and each agent $i$ updates its $u_i$ by
%\begin{align}\label{eq:quad_opt_u_i}
%    u_i(k) = u_i(k-1) - \alpha \cdot (QA_i)u_i(k-1),
%\end{align}
%then collectively,
%\begin{align*}
%    y(k) &= \sum_{i=1}^N A_{\D_{u,i}} u_i(k) \\
%        &= \sum_{i=1}^N \Big(A_{\D_{u,i}} u_i(k-1) - \alpha\cdot (Q A_{\D_{u,i}})u_i(k-1)\Big) \\
        %&= y(k-1) - \alpha \cdot Q y(k-1).
%\end{align*}
%The gradient descent method that solves~\eqref{eq:quad_opt} is written as 
%\begin{align}\label{eq:quad_opt_grad}
%    y(k) &= y(k-1) - \alpha Qy(k-1)  \\
%          &= y(k-1) - \alpha\cdot Q\Big(\sum_{i=1}^N A_i u_i(k-1)\Big).
%\end{align}
%If $Q$ is known for all the agents and the sub-matrices $A_i$ for all $i\in\Nc$ are identified with Algorithm~\ref{alg:D-ADAM}, \eqref{eq:quad_opt_grad} can be implemented in a decentralized way. 

\section{Numerical Studies}
We validate the proposed distributed model identification algorithm with two test cases. We first demonstrate that by adding Adam algorithm adaptive step-size, a faster convergence rate is observed for Algorithm~\ref{alg:D-ADAM} compared to~\eqref{eq:pd_dis}. We next apply the proposed algorithm to identify the LinDistFlow model~\cite{baran1989network} (or its equivalence) for a modified IEEE 37 buses system, and show a satisfactory control performance by leveraging the identified model. 
\subsection{A small-scale example}
We randomly generate a linear networked system given as
\begin{align}\label{eq:ex1}
    y= [y_1; y_2; \cdots, y_N] = \sum_{i=1}^N A_{\D_{u,i}} u_i,
\end{align}
where $u_i\in\real^{|\D_{u,i}|}$ and $y_i\in\real^{|\D_{y,i}|}$ are the input and output of agent $i$, $|\D_{u,i}|\in\{4,5\}$, $|\D_{y,i}|\in\{3,4\}$, and $N=5$. All the entries of matrix $A = [A_{\D_{u,1}},A_{\D_{u,2}},\cdots, A_{\D_{u,N}}]$ and $u_i$ are randomly generated with the standard normal distribution with the standard deviation of 1. We assume that the model identification algorithm~\eqref{eq:pd_dis} and Algorithm~\ref{alg:D-ADAM} start when the data matrices $U$ and $Y$ are constructed with $T$ steps, $T = 2\sum_{i=1}^N |\D_{u,i}|$. The algorithmic parameters are set as $\alpha = 10^{-3}$, $\beta_1 = 0.9$, $\beta_2 = 0.95$, and $\epsilon = 10^{-8}$. As shown in Figures~\ref{fig:Test_ID} and~\ref{fig:Test_ID_ADAM}, both~\eqref{eq:pd_dis} and Algorithm~\ref{alg:D-ADAM} have the estimated model converges to the actual ones.  Algorithm~\ref{alg:D-ADAM} has a noticeably  faster convergence rate, especially at the first hundreds iterations. We conjecture the dominating factor for the convergence after few hundred steps is on the consensus of $\zb$, so there is no much difference between~\eqref{eq:pd_dis} and Algorithm~\ref{alg:D-ADAM} afterward.
\begin{figure}
    \centering
    \includegraphics[width=0.88\linewidth]{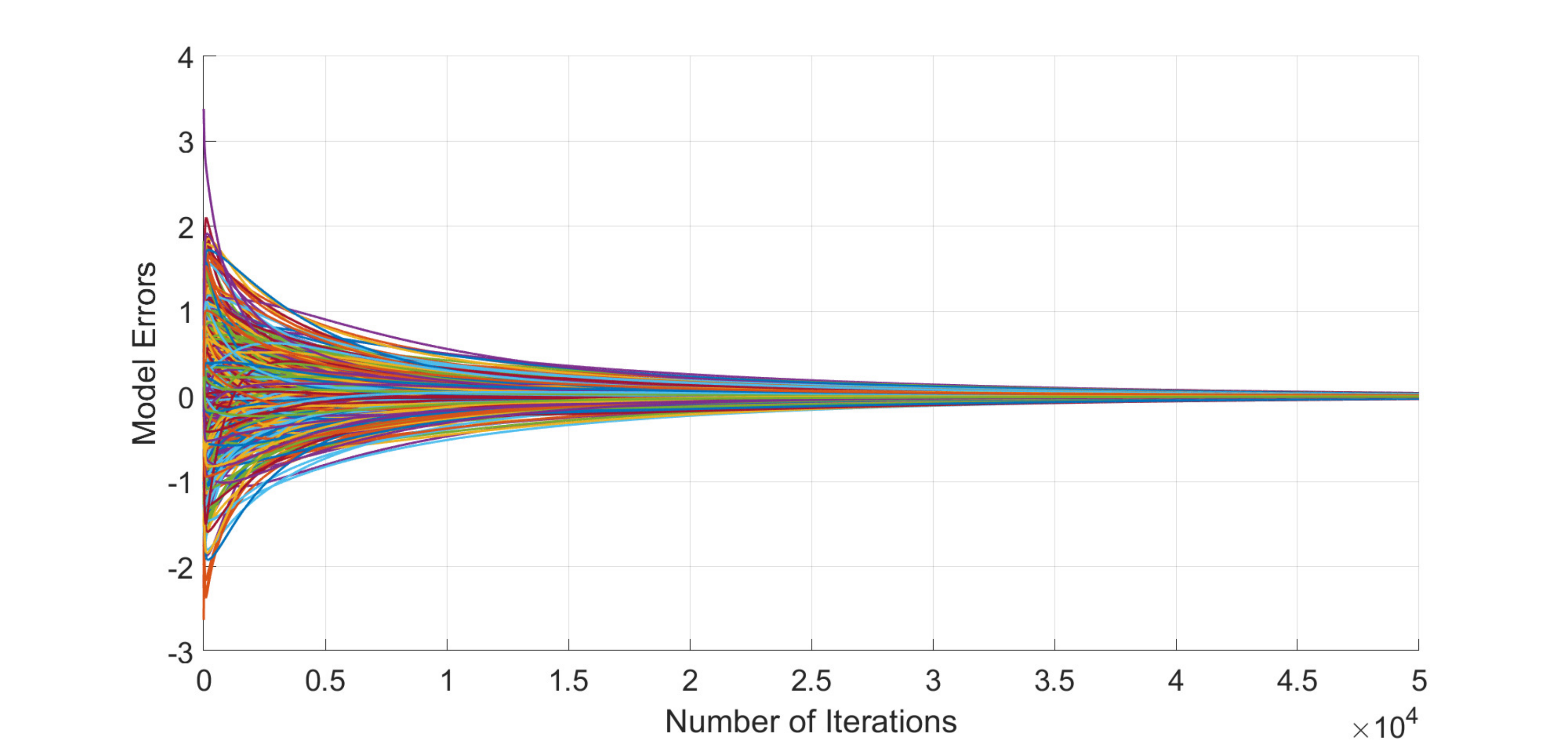}
    \caption{The absolute value difference for each element of the estimated $A$ and $A^\star$ by implementing~\eqref{eq:pd_dis}.}
    \label{fig:Test_ID}
\end{figure}

\begin{figure}
    \centering
    \includegraphics[width=0.88\linewidth]{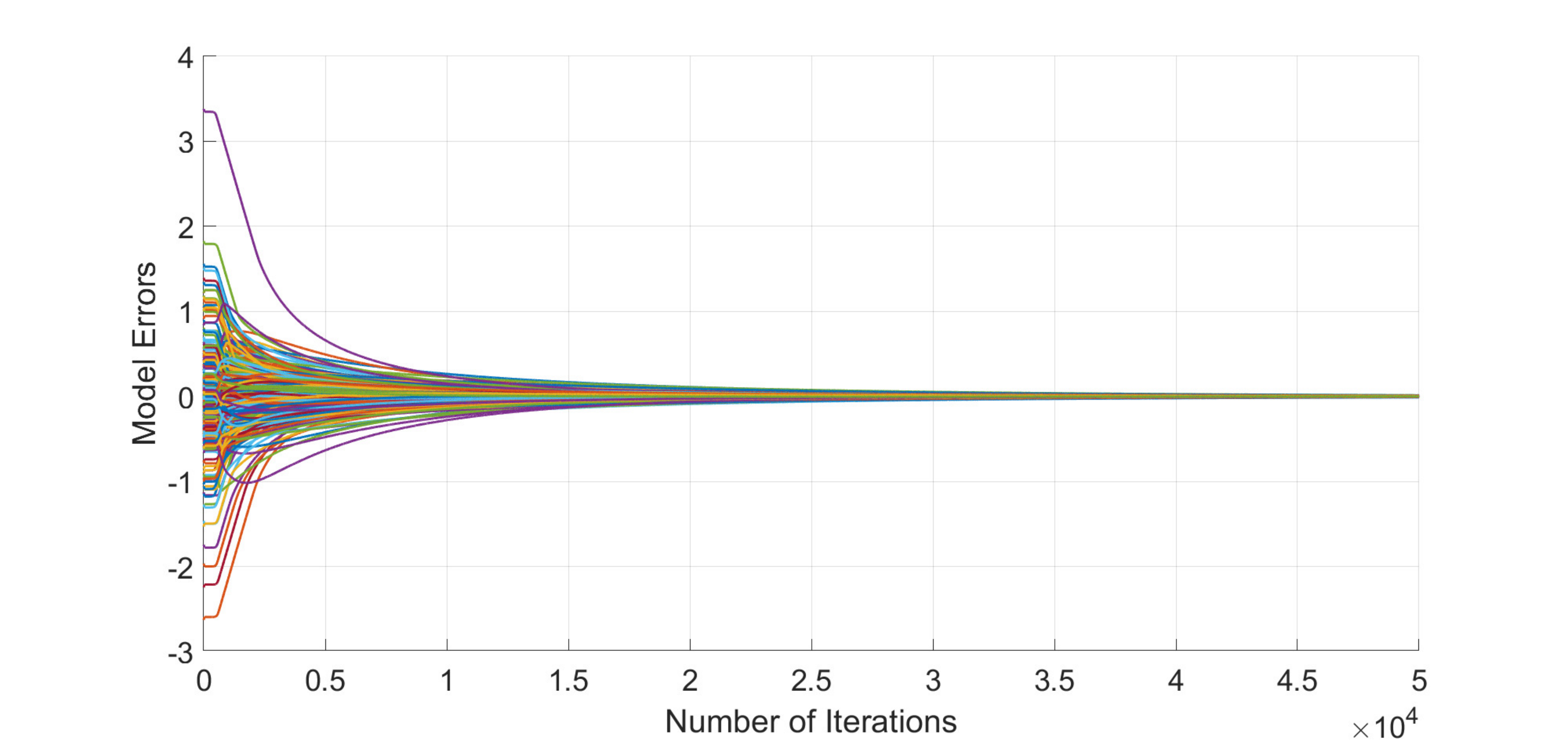}
    \caption{The absolute value difference for each element of the estimated $A$ and $A^\star$ for Algorithm~\ref{alg:D-ADAM}.}
    \label{fig:Test_ID_ADAM}
\end{figure}

\subsection{A modified IEEE 37-bus test system}
In this section, we apply Algorithm~\ref{alg:D-ADAM} to identify the LinDistFlow model $A$ for the voltage magnitude regulation purpose. We modify the IEEE 37-bus test system with penetration of PVs as illustrated in Figure~\ref{fig:IEEE37}. We assume that every PV bus acts as an agent and knows its local active/reactive power injections and the voltage magnitude of the bus, and it can adjust its power injections to help regulate the voltage magnitudes. The power injections of PV bus $i$ is collected as the control variable $u_i\in\real^2$ and the voltage magnitude measurements are collected by $y_i\in\real_+$. We assume that a connected distributed communication network $\graph$ is established for the purposes of running Algorithm~\ref{alg:D-ADAM} and the distributed feedback-based algorithm developed in~\cite{chang2019saddle}. 
%The PV buses for implementing Algorithm~\ref{alg:D-ADAM} to identifying the sub-models that draws relations between $[u_1,u_2,\cdots,u_N]$ and $[y_1,y_2,\cdots,y_N]$, and the sub-models are used for the distributed feedback-based control algorithm developed in~\cite{chang2019saddle}. 
\begin{figure}
    \centering
    \includegraphics[width=0.7\linewidth]{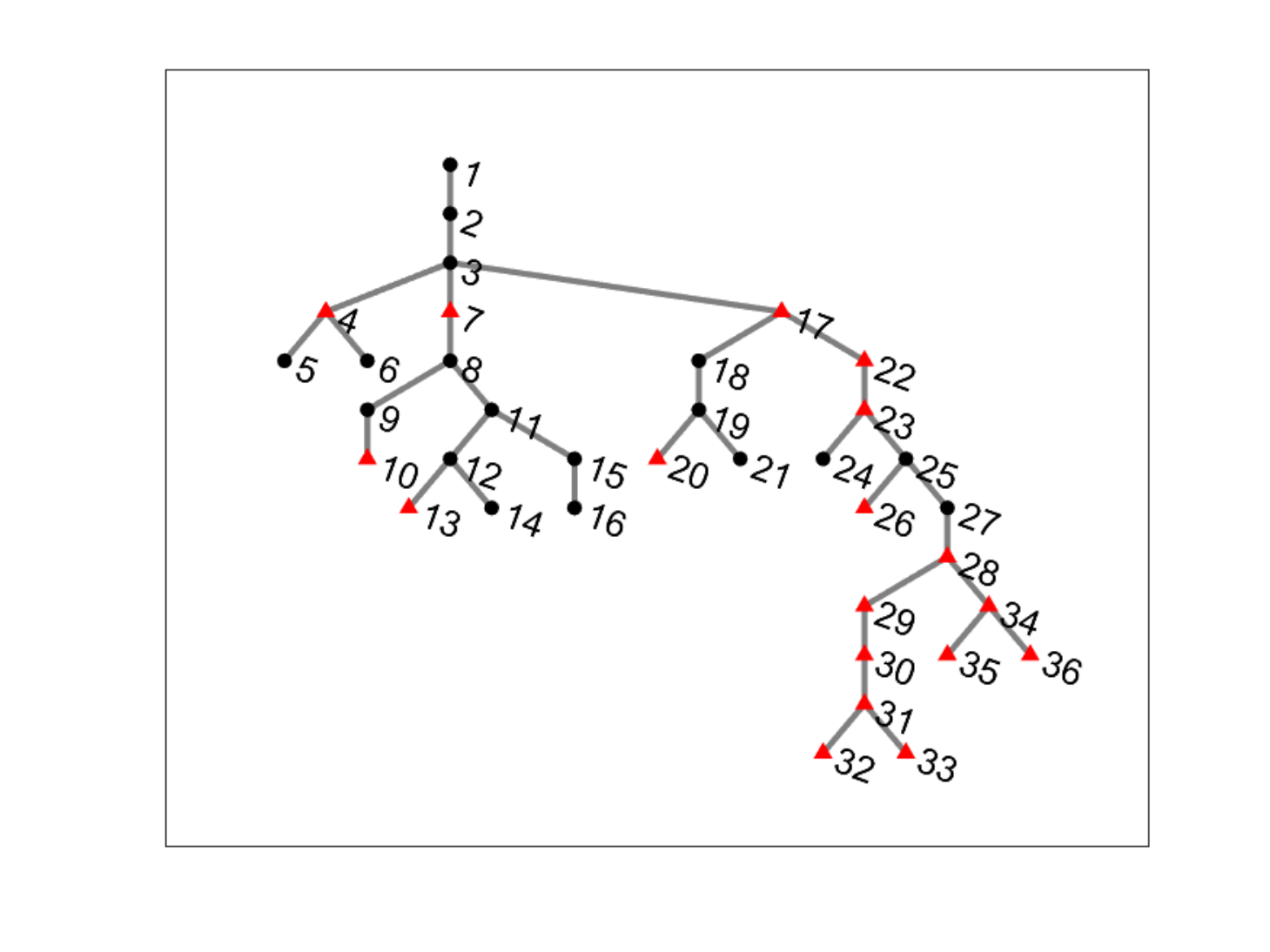}
    \caption{Illustration of the modified IEEE 37-bus system with the buses highlighted in red triangles are PV buses}
    \label{fig:IEEE37}
\end{figure}

Implementing the distributed feedback-based control algorithm developed in~\cite{chang2019saddle} requires a LinDistFlow model $A$ or its equivalence. Although the model $A$ does not need to be very accurate as explained in~\cite{colombino2019towards}, a reasonably accurate $A$ is still needed. We simulate the following scenarios: (i) no control on the voltage magnitudes; (ii) distributed feedback-based control with $A$ derived by the given knowledge of line impedances and network topology; (iii) distributed feedback-based control with $A$ identified by Algorithm~\ref{alg:D-ADAM} offline; (iv) distributed feedback-based control with $A$ identified by Algorithm~\ref{alg:D-ADAM} online whenever the new pair of $u$ and $y$ is available to replace the oldest pair. The algorithmic parameters $\alpha$, $\beta_1$, $\beta_2$ and $\epsilon$ are set the same as the first numerical example for Algorithm~\ref{alg:D-ADAM}. The number of time instances of constructing input and output data matrices, $U$ and $Y$, is set by $T = 140$.   

Figures~\ref{fig:Vwo}-\ref{fig:Vw_est_online} illustrate the voltage magnitudes for the four scenarios. Unsurprisingly, the control performance with the known $A$ is among the bests. An interesting observation is that for scenario (iii), the $A$ concatenated by the identified $A_{\D_{u,i}}$ for all the agents is very different from the $A$ derived from the given impedance and topology, regardless how we initialize $A_{\D_{u,i}}$ (or $\xb_i$). However, a similar voltage regulation result to scenario (ii) is achieved. Our explanation is that another form of linearized power flow is identified through Algorithm~\ref{alg:D-ADAM} because there can be many valid linearized models. Scenario (iv) is expected to perform best because it adapts the  model $A$ in accordance with the operating points online (could be understood as adjusting the linearization point). However, as shown in Figure~\ref{fig:Vw_est_online}, we observe a bit more fluctuation of the regulated voltage magnitudes compared to scenarios (ii) and (iii). Some fine-tuning of the number of steps before updating $U$ and $Y$ may be needed for a better performance. Overall, both scenarios (iii)-(iv) achieve satisfactory voltage regulations while keeping the power injections and voltage magnitude measurements local, and the packet exchanges are limited to dual variables associated with voltage constraints and dummy variables for the distributed algorithms. 
\begin{figure}
    \centering
    \includegraphics[width=0.88\linewidth]{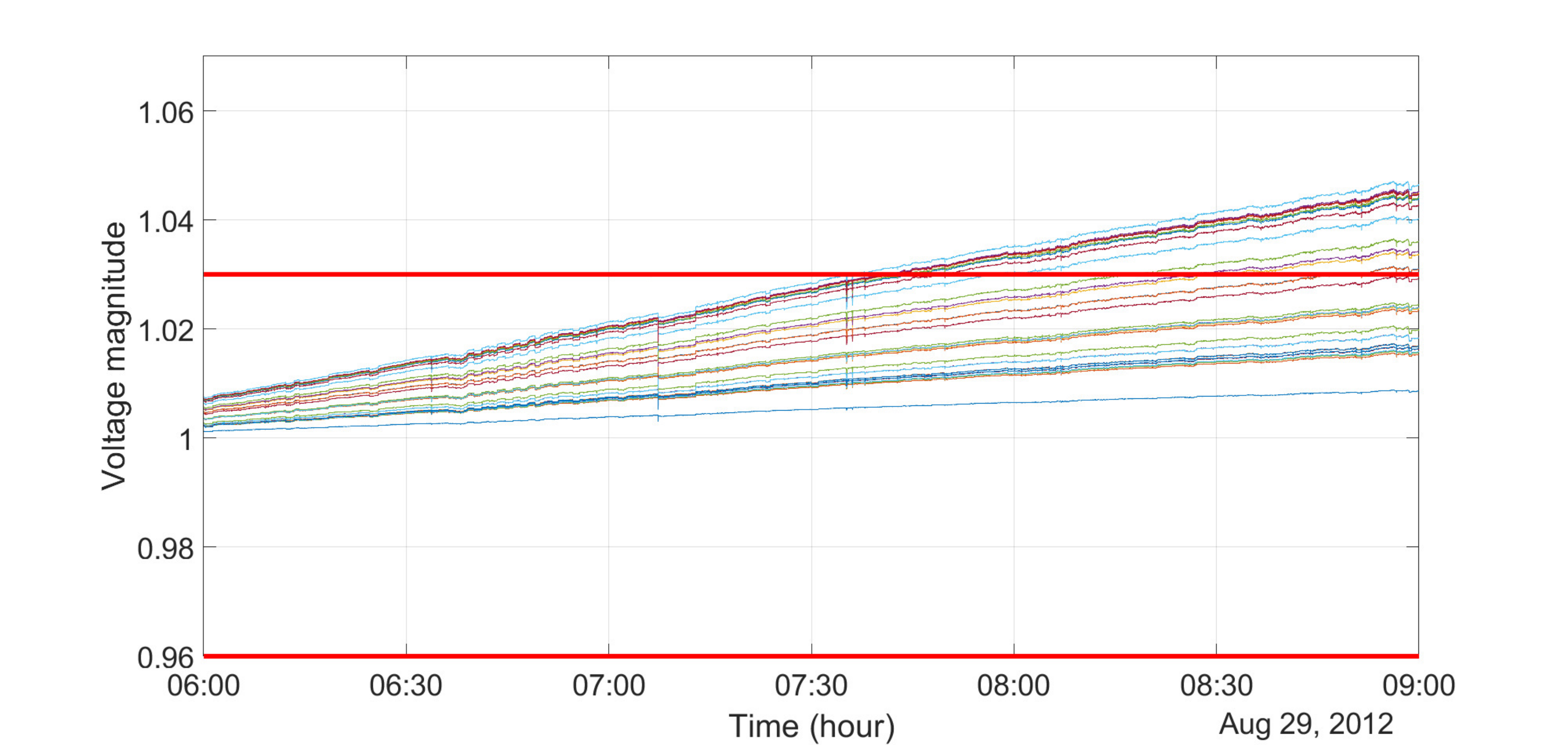}
    \caption{The voltage magnitudes (p.u.) over time without control.}
    \label{fig:Vwo}
\end{figure}

\begin{figure}
    \centering
    \includegraphics[width=0.88\linewidth]{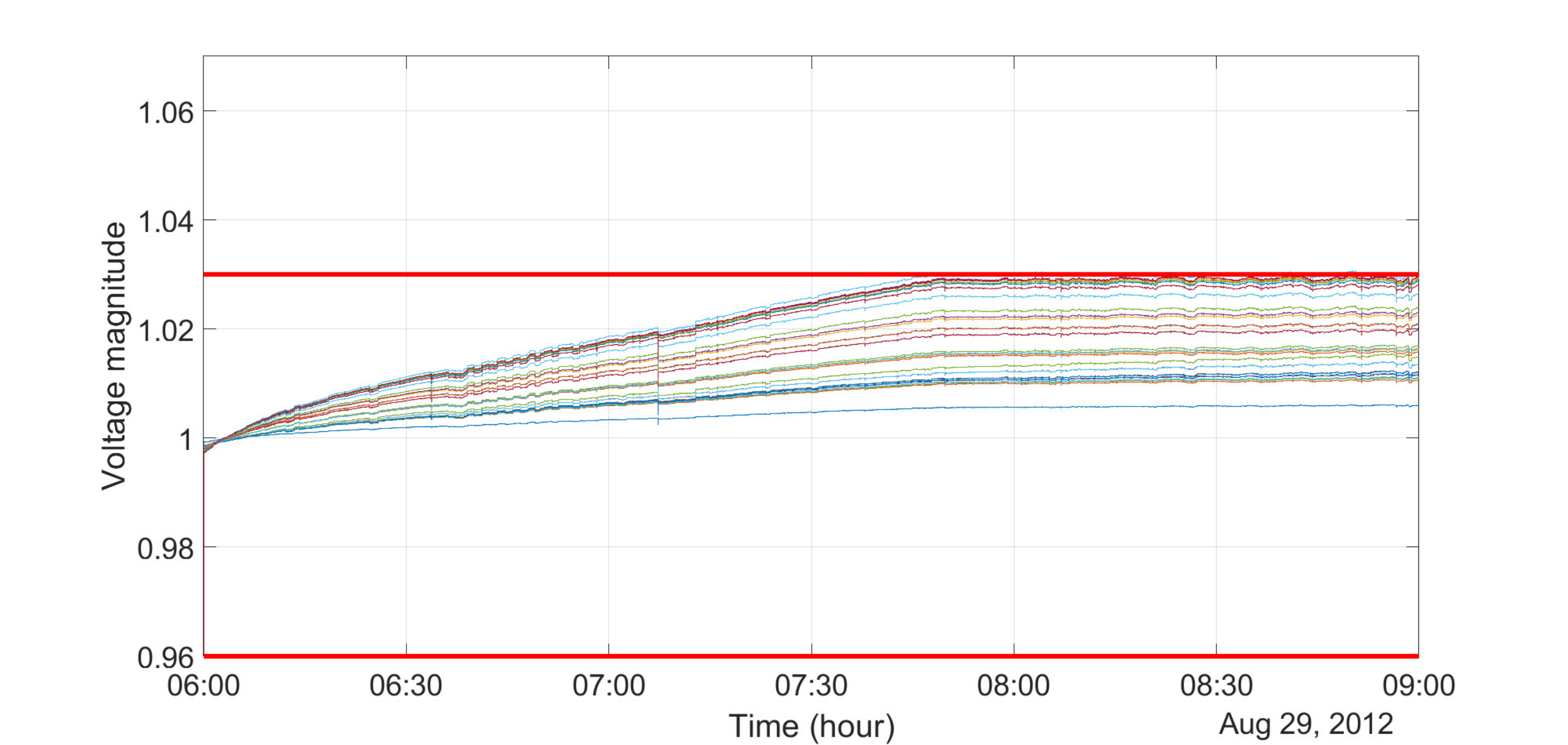}
    \caption{The voltage magnitudes (p.u.) over time with distributed feedback-based control and known LinDistFlow model.}
    \label{fig:Vw}
\end{figure}

\begin{figure}
    \centering
    \includegraphics[width=0.88\linewidth]{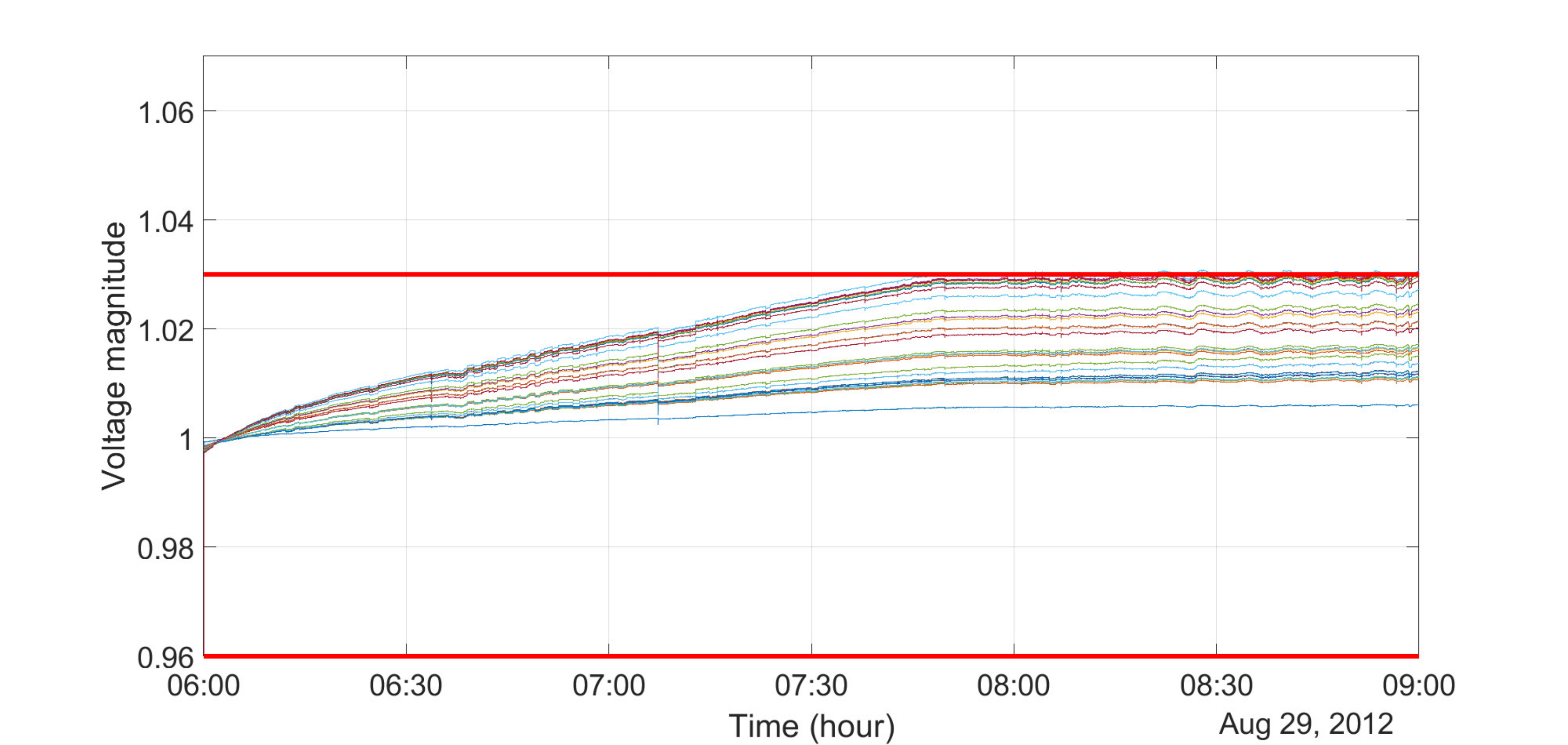}
    \caption{The voltage magnitudes (p.u.) over time with distributed feedback-based control and the identified LinDistFlow model.}
    \label{fig:Vw_est}
\end{figure}

\begin{figure}
    \centering
    \includegraphics[width=0.88\linewidth]{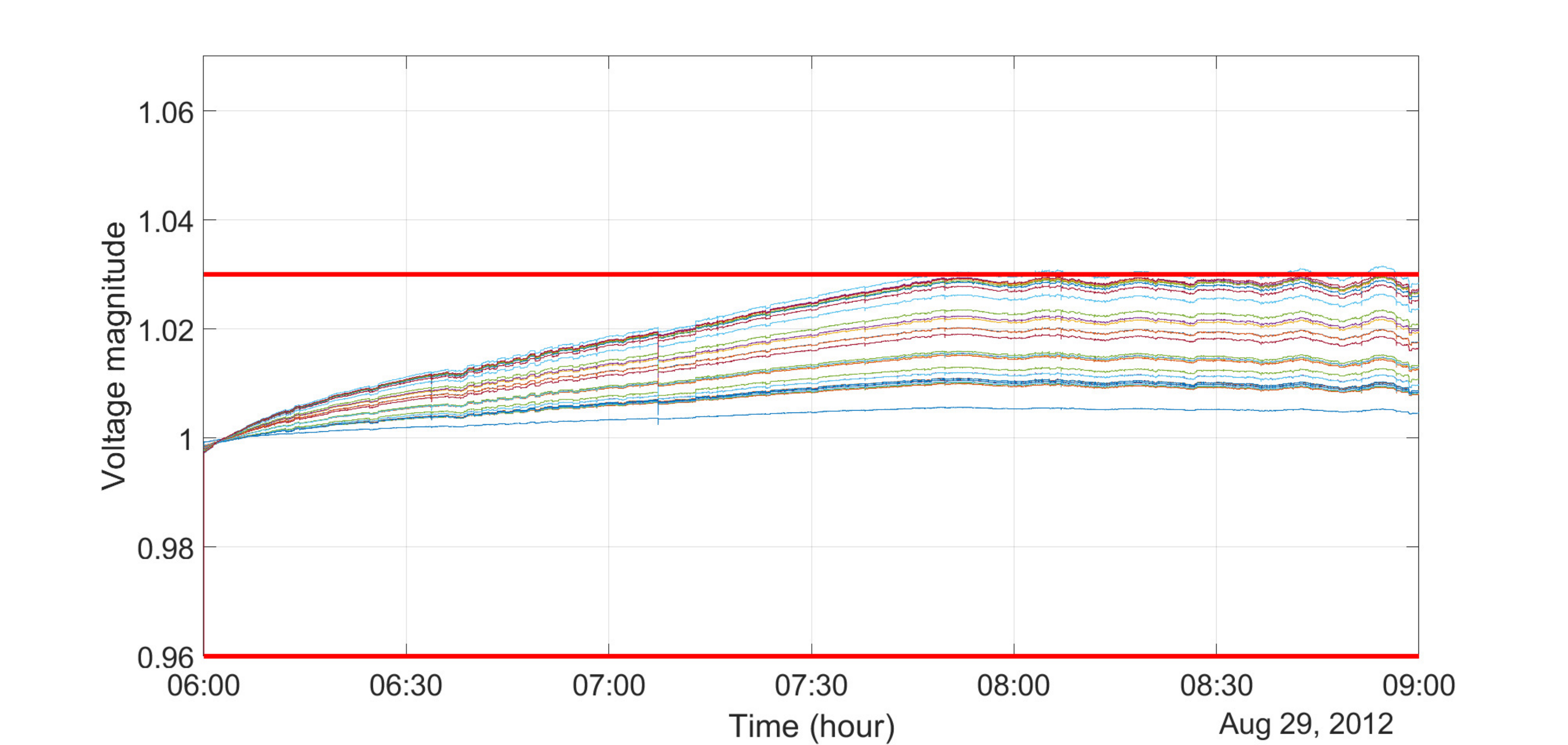}
    \caption{The voltage magnitudes (p.u.) over time with distributed feedback-based control and the identified LinDistFlow model updated online.}
    \label{fig:Vw_est_online}
\end{figure}

% Collecting the data from all the actuators and sensors can only happen offline; however, it is more desirable to update the model dynamically when new data come in so that the system model is up-to-date for controller designs. The centralized collection of all the actuators and sensors data in real time is unrealistic, which motivates us to develop a distributed method that dynamically update the system model in real time. The proposed distributed model identification algorithm can be a part of ``model-free'' controller design in the sense that the model is straight from locally available data. 

\section{Conclusion}
In this paper, we proposed a distributed model identification algorithm such that each agent identifies a sub-model that describes its local controls with the overall system outputs. The algorithm is designed such that the agents do not need to share their local data and the convergence rate is practically better than the gradient descent consensus algorithm because of adding the adaptive step-size of the Adam algorithm. We focus on the power distribution system applications in this paper, but we envision the proposed algorithm is potentially useful for decentralized or distributed controls in many other networked systems such as robotics, economics, telecommunications. Our near future works will be on refining the online distributed model identification for improved linearized models for nonlinear or time-varying linear systems. More in-depth testing of the proposed algorithm for some potential applications is also among our future works.

%In fact, the distributed algorithm~\eqref{eq:pd} is similar to the training process of machine learning in the sense that only a batch of data is used at a time and those data are used alternatively for the training. The difference is that in our distributed model identification algorithm~\eqref{eq:pd}, the ``training'' is executed distributively and the goal is identifying the actual model for the system. 

%Second, built on the identified sub-models, we develop a gradient-based decentralized controller that is equivalent to the centralized gradient control. There can be various forms of decentralized controllers built on the identified sub-models depending on the needs of the applications.

%The proposed method only applicable for system with certain model structures. The model may not capture the entire correlation between the inputs and outputs. Another potential future direction is combining the distributed control with the controls generated by a neural network model with latent data collected by a central agent. 

\bibliographystyle{IEEEtran}
\bibliography{biblio.bib}
\end{document}